\def\dayte{1}
\def\arXiv{}

\ifdefined\arXiv
\documentclass[a4paper]{article}
\else
\documentclass[10pt,conference]{IEEEtran}
\fi

\usepackage{amsmath,amssymb,amsthm}
\usepackage{authblk}
\usepackage{scalerel}
\usepackage{todonotes}
\usepackage{xspace}
\usepackage{diagbox}
\usepackage[left=2.5cm,right=2.5cm, top=2cm, bottom=3cm]{geometry}
\usepackage{tikz}
\usepackage{circuitikz}
\usetikzlibrary{automata,positioning}
\usetikzlibrary{circuits.logic.US} 
\usetikzlibrary{circuits,calc}

\usepackage{balance}

\usepackage{comment}
\usepackage{paralist}
\usepackage{multirow}
\usepackage{adjustbox}
\usepackage{epstopdf}

\usepackage{listings}
\ifdefined\arXiv
\else
\usepackage{color}
\usepackage{soul}
\newcommand{\moti}[1]{\hl{\textbf{M}: #1}}

\fi

\def\nopfs{0} 
\ifnum\nopfs=1
\usepackage{environ}
\NewEnviron{killcontents}{}

\fi

\newtheorem{theorem}{Theorem}[section]

\newtheorem{definition}[theorem]{Definition}

\newtheorem{lemma}[theorem]{Lemma}
\newtheorem{observation}[theorem]{Observation}

\newcommand{\NN}{\mathbb{N}}

\newcommand{\outt}{\mathrm{out}}

\newcommand{\rg}{\operatorname{rg}}

\newcommand{\ANDD}{\operatorname{AND}}
\newcommand{\AOI}{\operatorname{AOI}}

\newcommand{\ORR}{\operatorname{OR}}

\newcommand{\XORR}{\operatorname{XOR}}
\newcommand{\rgmax}{\operatorname{max^{\rg}}}
\newcommand{\rgmin}{\operatorname{min^{\rg}}}
\newcommand{\rgmaxM}{\operatorname{max^{\rg}_{\metas}}}
\newcommand{\rgminM}{\operatorname{min^{\rg}_{\metas}}}

\newcommand{\cmux}{\textsc{cmux}}
\newcommand{\mux}{\textsc{mux}}

\newcommand{\validrg}[1]{\mathcal{S}^{#1}_{\rg}}

\newcommand{\metas}{\textnormal{\texttt{M}}}

\newcommand{\twosort}{\operatorname{2-sort}}

\newcommand{\foursort}{\operatorname{4-sort}}
\newcommand{\sevensort}{\operatorname{7-sort}}
\newcommand{\tensort}{\operatorname{10-sort}}
\newcommand{\tensortc}{\operatorname{10-sort_\#}}
\newcommand{\tensortd}{\operatorname{10-sort_d}}

\newcommand{\parity}{\operatorname{par}}
\newcommand{\res}{\operatorname{res}}

\newcommand{\binary}{\operatorname{Bin-comp}}

\usepackage{amsmath}
\usepackage{relsize}
\newcommand{\bigstarr}{\mathop{\raisebox{-.7pt}{\ensuremath{\mathlarger{\mathlarger{\mathlarger{*}}}}}}}
\DeclareMathOperator*{\bigdiamond}{\scalerel*{\diamond}{\textstyle\sum}}
\DeclareMathOperator*{\hatdiamM}{\,\hat{\diamond}_{\metas}\,}

\newcommand{\delay}{\operatorname{delay}}
\newcommand{\cost}{\operatorname{cost}}

\newcommand{\repg}[1]{\langle #1 \rangle}

\begin{document}

\title{Optimal Metastability-Containing Sorting Networks}
\date{}
\author[1]{Johannes~Bund}
\author[1]{Christoph~Lenzen}
\author[2]{Moti~Medina}
\affil[1]{Saarland Informatics Campus, MPI for Informatics, Germany,
          Email:~\texttt{\{jbund,clenzen\}@mpi-inf.mpg.de}}
\affil[2]{Dept. of Electrical \& Computer Engineering,
          Ben-Gurion~University~of~the~Negev, Israel,
          Email:~\texttt{medinamo@bgu.ac.il}}

\maketitle

\begin{abstract}
When setup/hold times of bistable elements are violated, they may become
metastable, i.e., enter a transient state that is neither digital 0 nor
1~\cite{Mar81}. In general, metastability cannot be avoided, a problem that
manifests whenever taking discrete measurements of analog values. Metastability
of the output then reflects uncertainty as to whether a measurement should be
rounded up or down to the next possible measurement outcome.

Surprisingly, Lenzen \& Medina (ASYNC 2016) showed that metastability can be
\emph{contained}, i.e., measurement values can be correctly sorted
\emph{without} resolving metastability first. However, both their work and the
state of the art by Bund et al.\ (DATE 2017) leave open whether such a solution
can be as small and fast as standard sorting networks. We show that this is
indeed possible, by providing a circuit that sorts Gray code inputs (possibly
containing a metastable bit) and has asymptotically optimal depth and size.
Concretely, for $10$-channel sorting networks and $16$-bit wide inputs, we
improve by $48.46\%$ in delay and by $71.58\%$ in area over Bund et al. Our
simulations indicate that straightforward transistor-level optimization is
likely to result in performance on par with standard (non-containing) solutions.
\end{abstract}


\section{Introduction}\label{sec:intro}

Metastability is one of the basic obstacles when crossing clock domains,
potentially resulting in soft errors with critical consequences~\cite{ginosar11tutorial}. As it has
been shown that there is no deterministic way of avoiding metastability~\cite{Mar81},
synchronizers~\cite{kinniment08} are employed to reduce the error probability to
tolerable levels. Besides energy and chip area, this approach costs time: the
more time is allocated for metastability resolution, the smaller is the
probability of a (possibly devastating) metastability-induced fault.

Recently, a different approach has been proposed, coined
\emph{metastability-containing} (MC) circuits~\cite{friedrichs16}. The idea is
to accept (a limited amount of) metastability in the input to a digital circuit
and guarantee limited metastability of its output, such that the result is still
useful. The authors of~\cite{date17,async16} apply this approach to a
fundamental primitive: sorting. However, the state-of-the-art~\cite{date17} are
circuits that are by a $\Theta(\log B)$ factor larger than non-containing
solutions, where $B$ is the bit width of inputs. Accordingly, the authors pose
the following question:
\begin{quote}
\centering
\textit{``What is the optimum cost of the $\twosort$ primitive?''}\hspace*{-.105cm}
\end{quote}
We argue that answering this question is critical, as the performance penalty
imposed by current MC sorting primitives is not outweighed by the avoidance of
synchronizers.

\paragraph*{Our Contribution}
We answer the above question by providing a $B$-bit MC $\twosort$ circuit of
depth $O(\log B)$ and $O(B)$ gates. Trivially, any such building block with gates of
constant fan-in must have this asymptotic depth and gate count, and it improves
by a factor of $\Theta(\log B)$ on the gate complexity of~\cite{date17}.
Furthermore, we provide optimized building blocks that significantly improve the
leading constants of these complexity bounds. See Figure~\ref{fig:areaplot} for
our improvements over prior work; specifically, for $16$-bit inputs, area and
delay decrease by up to $71.58\%$ and $48.46\%$ respectively.

Plugging our circuit into (optimal depth or size) sorting
networks~\cite{bundala2014optimal,codish2014twenty,knuth1998art}, we obtain
efficient combinational metastability-containing sorting circuits,
cf.\ Table~\ref{table:sorting}. In general, plugging our $\twosort$ circuit into
an $n$-channel sorting network of depth $O(\log n)$ with $O(n \log
n)$ $\twosort$ elements~\cite{ajtai83}, we obtain an asymptotically optimal MC
sorting network of depth $O(\log B \log n)$ and $O(B n \log n)$ gates.

\begin{figure}
\centering
\includegraphics[width=.5\columnwidth]{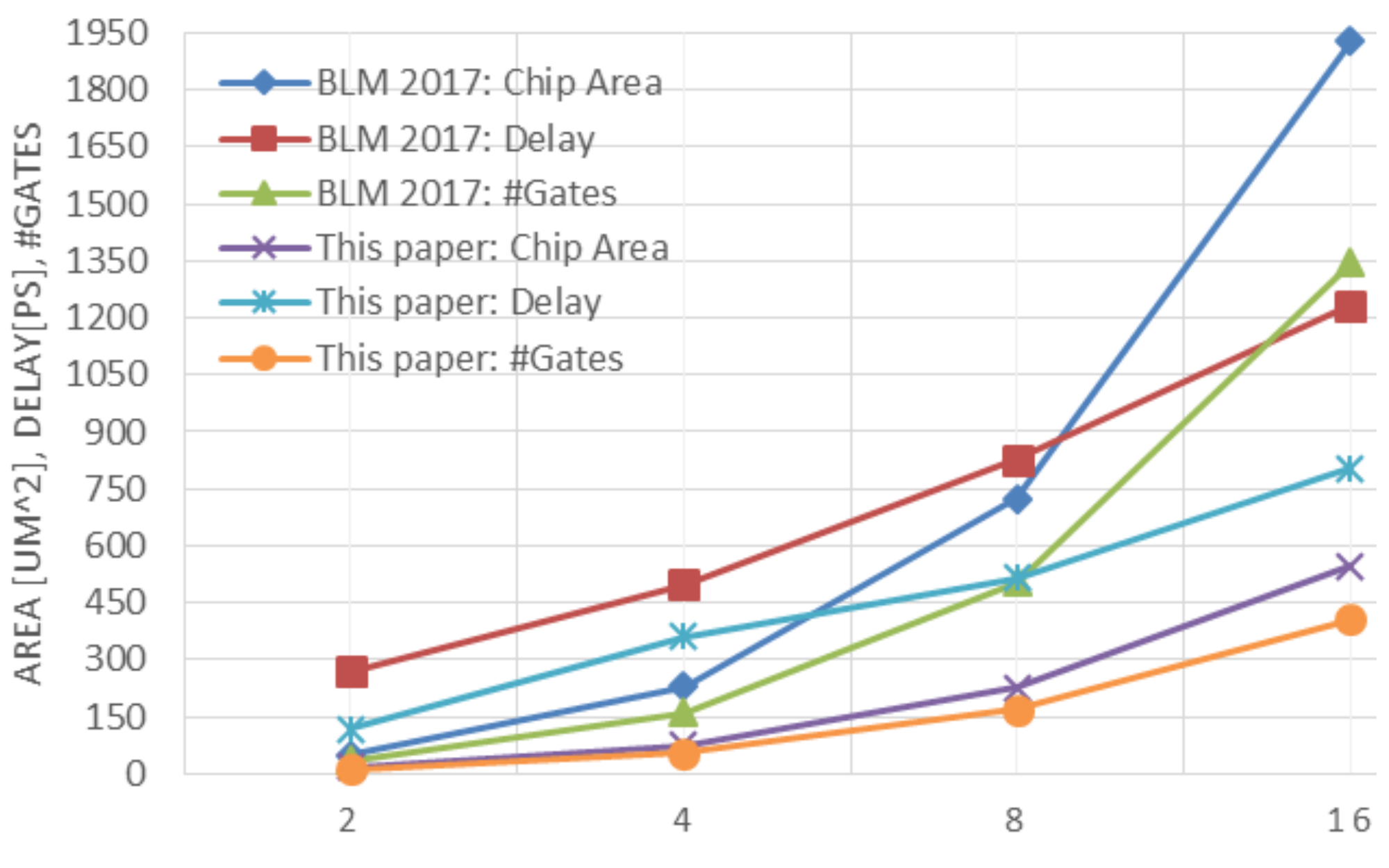}
\caption{Area, delay, and gate count of $\twosort(B)$ for $B \in
\{2,4,8,16\}$; we compare our solution to~\cite{date17}.}\label{fig:areaplot}
\end{figure}

\paragraph*{Further Related Work}
Ladner and Fischer~\cite{ladner1980parallel} studied the problem of computing
all the prefixes of applications of an associative operator on an input string
of length $n$. They designed and analyze a recursive construction which computes
all these prefixes in parallel. The resulting parallel prefix computation (PPC)
circuit has depth of $O(\log n)$ and gate count of $O(n)$ (assuming that the
implementation of the associative operator has constant size and constant
depth). We make use of their construction as part of ours.

\section{Model and Problem}\label{sec:model}

In this section, we discuss how to model metastability in a worst-case fashion
and formally specify the input/output behavior of our circuits.

We use the following basic notation. For $N\in \mathbb{N}$, we set $[N]:=
\{0,\ldots,N-1\}$. For a binary $B$-bit string $g$, denote by $g_i$ its $i$-th
bit, i.e., $g=g_1g_2\ldots g_B$. We use the shorthand $g_{i,j}:=g_i\ldots g_j$.
Let $\parity(g)$ denote the parity of $g$, i.e, $\parity(g) =
\sum_{i=1}^{B}g_i\bmod 2$.

\paragraph*{Reflected Binary Gray Code}

Due to possible metastability of inputs, we use Gray code. Denote by
$\repg{\cdot}$ the decoding function of a Gray code string, i.e., for $x\in
[N]$, $\repg{\rg_B(x)}=x$. As each $B$-bit string is a codeword, the code is a
bijection and the decoding function also defines the encoding function
$\rg_B:[N]\to \{0,1\}^B$.
We define $B$-bit binary reflected Gray code recursively, where a $1$-bit code
is given by $\rg_1(0)=0$ and $\rg_1(1)=1$. For $B>1$, we start with the first
bit fixed to $0$ and counting with $\rg_{B-1}(\cdot)$ (for the first $2^{B-1}-1$
codewords), then toggle the first bit to $1$, and finally ``count down''
$\rg_{B-1}(\cdot)$ while fixing the first bit again,
cf.~Table~\ref{table:graystruct}. Formally, this yields
\begin{equation*}
\rg_B(x):=\begin{cases}
0\rg_{B-1}(x)& \mbox{if }x\in [2^{B-1}]\\
1\rg_{B-1}(2^B-1-x)& \mbox{if }x\in [2^B]\setminus [2^{B-1}].
\end{cases}
\end{equation*}

\begin{table}
\begin{center}
\begin{tabular}{| c   c| c   c  || c   c  | c   c |}
\hline
 $\#$ & $g_1,g_{2,4}$ & $\#$ & $g_1,g_{2,4}$ & $\#$ & $g_1,g_{2,4}$ & $\#$ & $g_1,g_{2,4}$\tabularnewline
\hline
\hline
$0$           & $0,000$ &$4$           & $0,110$ & $8$           & $1,100$ & $12$           & $1,010$\tabularnewline
$1$           & $0,001$ & $5$           & $0,111$ & $9$           & $1,101$ & $13$           & $1,011$ \tabularnewline
$2$           & $0,011$ & $6$           & $0,101$ & $10$           & $1,111$ & $14$           & $1,001$ \tabularnewline
$3$           & $0,010$ & $7$           & $0,100$ & $11$           & $1,110$ & $15$           & $1,000$ \tabularnewline
\hline
\end{tabular}
\end{center}
\caption{4-bit binary reflected Gray code}
\label{table:graystruct}
\end{table}

We define the maximum and minimum of two binary reflected Gray code strings, $\rgmax$ and $\rgmin$ respectively,  in the usual way, as follows. For two binary reflected Gray
code strings $g,h\in \{0,1\}^B$, $\rgmax$ and $\rgmin$ are defined as
\begin{align*}
\left(\rgmax\{g,h\},\rgmin\{g,h\}\right)&:=\begin{cases}
(g,h) & \mbox{if }\repg{g}\geq \repg{h}\\
(h,g) & \mbox{if }\repg{g}\leq \repg{h}.
\end{cases}
\end{align*}

\paragraph*{Valid Strings}
In~\cite{async16}, the authors represent metastable ``bits'' by $\metas$. The
inputs to the sorting circuit may have some metastable bits, which means that
the respective signals behave out-of-spec from the perspective of Boolean logic.
Such inputs, referred to as \emph{valid strings}, are introduced with the help
of the following operator.
\begin{definition}[The $*$ operator~\cite{async16}]\label{def:star}
For $B\in \NN$, define the operator
$*: \{0,1,\metas\}^B \times \{0,1,\metas\}^B \rightarrow \{0,1,\metas\}^B$ by
\begin{equation*}
\forall i\in \{1,\ldots,B\}:(x * y)_i := \begin{cases}
x_i & \mbox{if }x_i=y_i\\
\metas & \mbox{else.}
\end{cases}
\end{equation*}
\end{definition}
\begin{observation}
The operator $*$ is associative and commutative. Hence, for a set
$S=\{x^{(1)},\ldots,x^{(k)}\}$ of $B$-bit strings, we can use the shorthand
$$\bigstarr S := \bigstarr_{x\in S} x := x^{(1)}*x^{(2)}*\ldots*x^{(k)}.$$
We call $\bigstarr S$ \emph{the superposition of the strings in $S$.}
\end{observation}

Valid strings have at most one metastable bit. If this bit resolves to either
$0$ or $1$, the resulting string encodes either $x$ or $x+1$ for some $x$,
cf.~Table~\ref{table:validinputs}.
\begin{definition}[Valid Strings~\cite{async16}]\label{def:validinput}
Let $B\in \mathbb{N}$ and $N=2^B$. Then, the set of \emph{valid strings of
length $B$} is
\begin{equation*}
\validrg{B}:=\rg_B([N])\cup \bigcup_{x\in [N-1]}
\{\rg_B(x)*\rg_B(x+1)\}\:,
\end{equation*}
where for a set $A$ we abbreviate $f(A):=\{f(y)\,|\,y\in A\}$.
\end{definition}
As pointed out in~\cite{date17}, inputs that are valid strings may, e.g., arise
from using suitable time-to-digital converters for measuring time
differences~\cite{tdc16}.

\begin{table}
\begin{center}
\begin{tabular}{| c | c || c  | c | c  |  c  || c | c |}
\hline
 $g$ & $\repg{g}$ & $g$ & $\repg{g}$ & $g$ & $\repg{g}$ & $g$ & $\repg{g}$\tabularnewline
\hline
\hline
$0000$&$0$   & $0110$ &$4$   & $1100$ &$8$   & $1010$ &$12$\tabularnewline
$000\metas$&$-$ & $011\metas$ &$-$ & $110\metas$ &$-$ &$101\metas$&$-$\tabularnewline
$0001$&$1$   & $0111$ &$5$   & $1101$ &$9$   & $1011$ &$13$\tabularnewline
$00\metas1$&$-$ & $01\metas1$ &$-$ & $11\metas1$ &$-$ &$10\metas1$&$-$\tabularnewline
$0011$&$2$   & $0101$ &$6$   & $1111$ &$10$  & $1001$ &$14$\tabularnewline
$001\metas$&$-$ & $010\metas$ &$-$ & $111\metas$&$-$&$100\metas$&$-$\tabularnewline
$0010$&$3$   & $0100$ &$7$   & $1110$ &$11$  & $1000$ &$15$\tabularnewline
$0\metas10$&$-$ & $\metas100$ &$-$ & $1\metas10$ &$-$& $-$ &$-$\tabularnewline
\hline
\end{tabular}
\end{center}
\caption{$4$-bit valid inputs}
\label{table:validinputs}
\end{table}
\ifnum\nopfs=0
\begin{observation}\label{obs:valid_substrings}
For any $1\leq i\leq j\leq B$ and $g\in \validrg{B}$, $g_{i,j}\in
\validrg{j-i+1}$, i.e., $g_{i,j}$ is a valid string, too.
\end{observation}
\fi
\begin{proof}
Follows immediately from Observation~\ref{obs:truncate}.
\end{proof}

\paragraph*{Resolution and Closure}
To extend the specification of $\rgmax$ and $\rgmin$ to valid strings, we
make use of the \emph{metastable closure}~\cite{friedrichs16}, which in turn
makes use of the \emph{resolution.}
\begin{definition}[Resolution~\cite{friedrichs16}]
For $x\in \{0,1,\metas\}^B$,
\begin{equation*}
\res(x):=\{y\!\in \!\{0,1\}^B|\forall i\in \{1,\ldots,B\}\colon x_i\neq
\metas\Rightarrow y_i=x_i\}.
\end{equation*}
\end{definition}
Thus, $\res(x)$ is the set of all strings obtained by replacing all $\metas$s in
$x$ by either $0$ or $1$: $\metas$ acts as a ``wild card.''\\
We note the following for later use.
\begin{observation} For any $x$, $\bigstarr \res(x)=x$. For any $S$,
$S\subseteq \res(\bigstarr S)$.
\end{observation}
The metastable
closure of an operator on binary inputs extends it to inputs that may
contain metastable bits. This is done by considering all resolutions of the
inputs, applying the operator, and taking the superposition of the results.
\begin{definition}[The $\metas$
Closure~\cite{friedrichs16}]\label{def:mcomp} Given an operator $f\colon
\{0,1\}^n\times \{0,1\}^n \to \{0,1\}^n$, its \emph{metastable closure
$f_{\metas}\colon \{0,1,\metas\}^n\times \{0,1,\metas\}^n \to \{0,1,\metas\}^n$}
is defined by
$$f_{\metas}(x):= \bigstarr f(\res(x)).$$
\end{definition}
\paragraph*{Output Specification}
We want to construct a circuit that outputs the maximum and minimum of two valid
strings, which will enable us to build sorting networks for valid strings.
First, however, we need to answer the question what it means to ask for the
maximum or minimum of valid strings. To this end, suppose a valid string is
$\rg_B(x)*\rg_B(x+1)$ for some $x\in [N-1]$, i.e., the string contains a
metastable bit that makes it uncertain whether the represented value is $x$ or
$x+1$. This means that the measurement the string represents was taken of a
value somewhere between $x$ and $x+1$. Moreover, if we wait for metastability to
resolve, the string will stabilize to either $\rg_B(x)$ or $\rg_B(x+1)$.
Accordingly, it makes sense to consider $\rg_B(x)*\rg_B(x+1)$ ``in between''
$\rg_B(x)$ and $\rg_B(x+1)$, resulting in the total order on valid strings given
by Table~\ref{table:validinputs}.\\
The above intuition can be formalized by extending $\rgmax$ and $\rgmin$ to
valid strings using the metastable closure.
\begin{definition}[\cite{date17,async16}]\label{def:twosort}
For $B\in \NN$, a $\twosort(B)$ circuit is specified as follows.
\begin{compactitem}
  \item\textbf{\emph{Input:}} $g,h \in \validrg{B}$\:,
  \item\textbf{\emph{Output:}} $g',h' \in \validrg{B}$\:,
  \item\textbf{\emph{Functionality:}} $g'=\rgmaxM\{g,h\}$,
  $h'=\rgminM\{g,h\}$.
\end{compactitem}
\end{definition}
As shown in~\cite{date17}, this definition indeed coincides with the one given
in~\cite{async16}, and for valid strings $g$ and $h$, $\rgmaxM\{g,h\}$ and
$\rgminM\{g,h\}$ are valid strings, too. More specifically, $\rgmaxM$ and
$\rgminM$ are the $\max$ and $\min$ operators w.r.t.\ the total order on valid
strings shown in Table~\ref{table:validinputs}, e.g.,
\begin{compactitem}
  \item $\rgmaxM\{1001,1000\}=\rg_4(15)=1000$,
  \item $\rgmaxM\{0\metas10,0010\}=\rg_4(3)*\rg_4(4)=0\metas10$,
  \item $\rgmaxM\{0\metas10,0110\}=\rg_4(4)=0110$.
\end{compactitem}

\paragraph*{Computational Model}
We seek to use standard components and combinational logic only. We use the
model of \cite{friedrichs16}, which specifies the behavior of basic gates on
metastable inputs via the metastable closure of their behavior on binary inputs.
For standard implementations of $\ANDD$ and $\ORR$ gates, this assumption is valid: if
$\metas$ represents an arbitrary, possibly time-dependent voltage between logical
$0$ and $1$, an $\ANDD$ gate will still output logical $0$ if the respective other
input is logical $0$. Similarly, an $\ORR$ gate with one input being logical $1$
suppresses metastability at the other input, cf.\ Table~\ref{tab:gates}.

As pointed out in~\cite{date17}, any additional reduction of metastability in
the output necessitates the use of non-combinational masking components (e.g.,
masking registers), analog components, and/or synchronizers, all of which are
outside of our computational model. Moreover, other than the usage of analog
components, these alternatives require to spend additional time, which we avoid
in this paper.

\begin{table}
\centering
\begin{tabular}{c|ccc}
  \diagbox[width=2.4em,height=2.4em]{b}{a} & 0 & 1 & \metas\\ \hline
  0 & 0 & 0 & 0\\
  1 & 0 & 1 & \metas\\
  \metas & 0 & \metas & \metas
\end{tabular}
\qquad
\begin{tabular}{c|ccc}
  \diagbox[width=2.4em,height=2.4em]{b}{a} & 0 & 1 & \metas\\ \hline
  0 & 0 & 1 & \metas\\
  1 & 1 & 1 & 1\\
  \metas & \metas & 1 & \metas
\end{tabular}
\qquad
\begin{tabular}{c|c}
a & $\bar{\text{a}}$\\ \hline
0 & 1\\
1 & 0\\
\metas & \metas
\end{tabular}
\caption{Logical extensions to metastable inputs of AND
(left), OR (center), and an inverter (right).}\label{tab:gates}
\end{table}

\section{Preliminaries on Stable Inputs}\label{sec:encoding}

\ifnum\nopfs=0
We note the following observation for later use. Informally, it states that
removing prefixes and suffixes from the code results in (repetition) of binary
reflected Gray codes.
\begin{observation}\label{obs:truncate}
For $B$-bit binary reflected Gray code, fix $1\leq i<j\leq B$, and consider
the sequence of strings obtained by (i) listing all codewords in ascending order
of encoded values, (ii) replacing each codeword $g$ by $g_{i,j}$, and (iii)
deleting all immediate repetitions (i.e., if two consecutive strings are
identical, keep only one of them). Then the resulting list repeatedly counts
``up'' and ``down'' through the codewords of $(j-i)$-bit binary reflected Gray
code.
\end{observation}
\fi
\begin{proof}
When removing the first bit of $B$-bit binary reflected Gray code, the claim
follows directly from the definition. By induction, we can confirm that the last
bit of $B$-bit code toggles on every second up-count, and
$\repg{g}=2\cdot\repg{g_{1,B-1}}+\XORR(\parity(g_{1,B-1}),g_B)$. Thus, the claim
holds if we either remove the first or last bit. As the same arguments apply
when we have a list counting ``up'' and ``down'' repeatedly, we can inductively
remove the first $i-1$ bits and the last $B-j$ bits to prove the general claim.
\end{proof}

\paragraph*{Comparing Stable Gray Code Strings via an FSM}
The following basic structural lemma leads to a straightforward way of comparing
binary reflected Gray code strings.
\begin{lemma}\label{lem:correctness_binary}
Let $g,h\in \{0,1\}^B$ such that $\repg{g}>\repg{h}$. Denote by $i\in
\{1,\ldots,B\}$ the first index such that $g_i\neq h_i$. Then $g_i=1$ (i.e.,
$h_i=0$) if $\parity(g_{1,i-1})=0$ and $g_i=0$ (i.e., $h_i=1$) if
$\parity(g_{1,i-1})=1$.
\end{lemma}
\begin{proof}
We prove the claim by induction on $B$, where the base case $B=1$ is trivial.
Now consider $B$-bit strings for some $B>1$ and assume that the claim holds for
$B-1$ bits. If $i=1$, again the claim trivially follows from the definition. If
$i>1$, we have that $g_1=h_1$. Denote $x=\repg{g}$ and $y=\repg{h}$. If
$g_1=h_1=0$, then $g_{2,B}=\rg_{B-1}(x)$ and $h_{2,B}=\rg_{B-1}(y)$. Thus, as
$x>y$ by assumption, the claim follows from the induction hypothesis. If
$g_1=h_1=1$, $g_{2,B}=\rg_{B-1}(2^B-1-x)$ and $h_{2,B}=\rg_{B-1}(2^B-1-y)$. Note
that $g':=h_{2,B}$ and $h':=g_{2,B}$ satisfy that $\repg{g'}>\repg{h'}$ and that
their first differing bit is $i-1$. By the induction hypothesis, we have that
$h'_{i-1}=g_i=1$ if
$\parity(g'_{1,i-2})=\parity(h'_{1,i-2})=\parity(g_{2,i-1})=1$ and,
accordingly, $g_i=0$ if $\parity(g_{2,i-1})=0$. As $g_1=1$,
$\parity(g_{1,i-1})=1-\parity(g_{2,i-1})$, and the claim follows.
\end{proof}

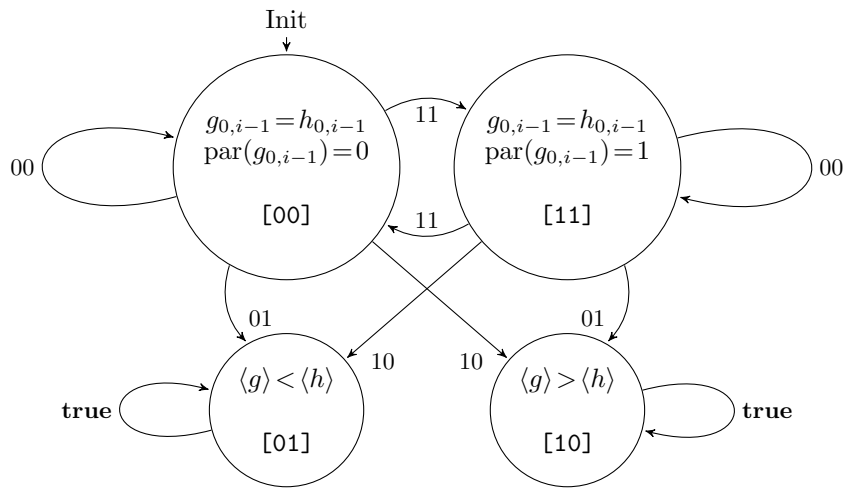
\begin{figure}
  \centering
\begin{tikzpicture}
[>=stealth',shorten >=1pt,auto,node distance=0.7cm, initial above,initial text=Init,scale=0.5]
   \node[state,initial,align=center] (00)   {$g_{0,i-1}\!=\!h_{0,i-1}$\\$\parity(g_{0,i-1})\!=\!0$\\\\\texttt{[00]}};
   \node[state] (11) [right=of 00,align=center] {$g_{0,i-1}\!=\!h_{0,i-1}$\\$\parity(g_{0,i-1})\!=\!1$\\\\\texttt{[11]}};
   \node[state] (01) [below =of 00,align=center] {$\repg{g}\!<\!\repg{h}$\\\\\texttt{[01]}};
   \node[state](10) [below =of 11,align=center] {$\repg{g}\!>\!\repg{h}$\\\\\texttt{[10]}};

   \begin{scope}[every node/.style={scale=.9}]
    \path[->]
    (00) edge [bend left,swap,align=center]node {$11$} (11)
          edge [loop left,align=center] node {$00$} ()
           edge  node [swap,very near end,align=center] {$10$}  (10)
             edge  [bend right,align=center,very near end] node {$01$} (01)
    (01) edge  [loop left] node  {$\textbf{true}$} ()
    (10) edge  [loop right] node  {$\textbf{true}$} ()
    (11) edge  [bend left,align=center] node [swap] {$11$} (00)
          edge  node [near start,align=center,very near end]  {$10$} (01)
           edge  [bend left,align=center,swap, very near end] node {$01$} (10)
           edge [loop right,align=center] node {$00$} ();
    \end{scope}
\end{tikzpicture}
  \caption{Finite state automaton determining which of two Gray code
  inputs $g,h\in \{0,1\}^B$ is larger. In each step, the machine receives
  $g_ih_i$ as input. State encoding is given in square brackets.}\label{fig:fsm}
\end{figure}

Lemma~\ref{lem:correctness_binary} gives rise to a sequential representation of
$\twosort(B)$ as a Finite state machine (FSM), for input strings in $\{0,1\}^B$.
Consider the state machine given in Figure~\ref{fig:fsm}. Its four states keep
track of whether $g_{1,i}=h_{1,i}$ with parity $0$ (state encoding: $00$) or $1$
(state encoding: $11$), respectively, $\repg{g}<\repg{h}$ (state encoding:
$01$), or $\repg{g}>\repg{h}$ (state encoding: $10$). Denoting by $s^{(i)}$ its
state after $i$ steps (where $s^{(0)}=00$ is the initial state),
Lemma~\ref{lem:correctness_binary} shows that the output given in
Table~\ref{tab:output} is correct: up to the first differing bits $g_i\neq h_i$,
the (identical) input bits are reproduced both for $\rgmax$ and $\rgmin$, and in
the $i$-th step the state machine transitions to the correct absorbing state.

\begin{table}
\centering
\begin{tabular}{|c||c |c |c |c |}
\hline \rule{0pt}{9pt}
  $s^{(i-1)}$ & $\rgmax\{g,h\}_i$ & $\rgmin\{g,h\}_i$\\ \hline \hline \rule{0pt}{8pt}
    	  00 & $\max\{g_i,h_i\}$ & $\min\{g_i,h_i\}$\\
		  10 & $g_i$ & $h_i$\\
  		  11 & $\min\{g_i,h_i\}$ & $\max\{g_i,h_i\}$\\
	      01 & $h_i$ & $g_i$\\
  \hline
\end{tabular}
\caption{Computing $\rgmax\{g,h\}_i$ and $\rgmin\{g,h\}_i$.}
\label{tab:output}
\end{table}

\paragraph*{The $\diamond$ Operator and Optimal Sorting of Stable Inputs}
We can express the transition function of the state machine as an operator
$\diamond$ taking the current state and input $g_ih_i$ as argument and returning
the new state. Then $s^{(i)}=s^{(i-1)}\diamond g_ih_i$, where $\diamond$ is
given in Table~\ref{tab:diamond}.

\begin{table}
\centering
\begin{tabular}{|c||c |c |c |c |}
\hline
  $\diamond$ & 00 & 01 & 11 & 10\\ \hline \hline
    	  00 & 00 & 01 & 11 & 10 \\
	      01 & 01 & 01 & 01 & 01 \\
  		  11 & 11 & 10 & 00 & 01 \\
		  10 & 10 & 10 & 10 & 10 \\
  \hline
\end{tabular}
\hspace{30pt}
\begin{tabular}{|c||c |c |c |c |}
\hline
  $\outt$ & 00 & 01 & 11 & 10\\ \hline \hline
       00 & 00 & 10 & 11 & 10 \\
	   01 & 00 & 10 & 11 & 01 \\
  	   11 & 00 & 01 & 11 & 01 \\
	   10 & 00 & 01 & 11 & 10 \\
  \hline
\end{tabular}
\caption{The $\diamond$ operator and the $\outt$ operator. The first operand is the current state, the
second the next input bits.}\label{tab:diamond}\label{tab:outt}
\end{table}

\begin{observation}\label{obs:associative}
$\diamond$ is associative, that is,
$$\forall a,b,c\in \{0,1\}^2\colon (a\diamond b)\diamond c=a\diamond (b\diamond c).$$
We thus have that
$$s^{(i)}=\bigdiamond_{j=1}^{i}g_jh_j:=g_1h_1\diamond g_2h_2\diamond\ldots\diamond g_ih_i,$$
regardless of the order in which the $\diamond$ operations are
applied.
\end{observation}
\begin{proof}
First, we observe the following for every $y \in \{0,1\}^2$: (1)~$00 \diamond y
= y$, (2)~$01 \diamond y = 01$, (3)~$11 \diamond y = \bar{y}$, and (4)~$10
\diamond y = 10$. We prove that $\diamond$ is associative by considering these
four cases for the first operand $x$. If $x \in \{01,10\}$, associativity
follows from the ``absorbing'' property of cases $(2)$ and $(4)$.
If $x=00$, then $(00 \diamond y) \diamond z = y \diamond z = 00 \diamond (y
\diamond z)$. We are left with the case that $x=11$. Then the LHS
equals $\bar{y} \diamond z$, while the RHS equals $\overline{y \diamond z}$.
Checking Table~\ref{tab:diamond}, one can directly verify that $\bar{y} \diamond
z=\overline{y \diamond z}$ in all cases.
\end{proof}

An immediate consequence is that we can apply the results
by~\cite{ladner1980parallel} on parallel prefix computation to derive an
$O(B)$-gate circuit of depth $O(\log B)$ computing all $s_i$, $i\in [B]$, in
parallel. Our goal in the following sections is to extend this well-known
approach to potentially metastable inputs.

\section{\sloppy Dealing with Metastable Inputs}\label{sec:optimal}

Our strategy is the same as outlined in Section~\ref{sec:encoding} for stable
inputs, where we replace all involved operators by their metastable closure:
(i) compute $s^{(i)}$ for $i\in [B]$, (ii) determine $\rgmax\{g,h\}_i$ and
$\rgmin\{g,h\}_i$ according to Table~\ref{tab:output} for $i\in \{1,\ldots,B\}$,
and (iii) exploit associativity of the operator computing the $s^{(i)}$ to
determine all of them concurrently with $O(\log B)$ depth and $O(B)$ gates
(using~\cite{ladner1980parallel}). To make this work for inputs that are valid
strings, we simply replace all involved operators by their respective metastable
closure. Thus, we only need to implement $\diamond_{\metas}$ and the
closure of the operator given in Table~\ref{tab:output} (both of
constant size) and immediately obtain an efficient circuit using the PPC
framework~\cite{ladner1980parallel}.

Unfortunately, it is not obvious that this approach yields correct outputs.
There are three hurdles to take:
\begin{enumerate}[(i)]
\item Show that first computing $s^{(i)}_{\metas}$ and then the output from this
and the input yields correct output for all valid strings.
\item Show that $\diamond_{\metas}$ behaves like an associative operator on the
given inputs (so we can use the PPC framework).
\item Show that repeated application of $\diamond_{\metas}$ actually computes
$s^{(i)}_{\metas}$.
\end{enumerate}

Killing two birds with one stone, we first show the second and third point in a
single inductive argument. We then proceed to prove the first point.

\subsection{Determining $s^{(i)}_{\metas}$}
Note that for any $x$ and $y$, we have that $\res(xy)=\res(x)\times \res(y)$.
Hence, for valid strings $g,h\in \validrg{B}$ and $i\in \{1,\ldots,B\}$, we have that
$$s^{(i)}_{\metas}=\bigstarr \bigdiamond_{j=1}^i \res(g_jh_j),$$
and for convenience set $s^{(0)}_{\metas}:=s^{(0)}=00$. Moreover, recalling
Definition~\ref{def:mcomp},
\begin{align}
x\diamond_{\metas} y = \bigstarr_{x'y'\in \res(xy)}\{x'\diamond y'\} = \bigstarr (\res(x)\diamond \res(y)).\label{eq:diamondM}
\end{align}
The following theorem shows that the desired decomposition is feasible.
\begin{theorem}\label{thm:decompose}
Let $g,h\in \validrg{B}$ and $1\leq i\leq j\leq B$. Then
\begin{align}\label{eq:claimed_equality}
g_ih_i\diamond_{\metas}g_{i+1}h_{i+1}
\diamond_{\metas}\ldots\diamond_{\metas}g_jh_j
&= \bigstarr \bigdiamond_{k=i}^j \res(g_kh_k),
\end{align}
regardless of the order in which the $\diamond_{\metas}$ operators are applied.
\end{theorem}
\begin{proof}
We start with a key observation.
\begin{observation}\label{obs:mm}
Let $g,h\in \validrg{B}$ and $1\leq i\leq j\leq B$. If
\begin{align*}
\bigstarr \bigdiamond_{k=i}^j \res(g_kh_k)=\metas\metas,
\end{align*}
there is an index $m$ such that $g_m=h_m=\metas$ and $g_{i,m}=h_{i,m}$.
Conversely, if there is no such index, then $x\neq \metas \metas$.
\end{observation}
\begin{proof}
Abbreviate $x=\bigstarr \bigdiamond_{k=i}^j \res(g_kh_k)$.
By Observation~\ref{obs:valid_substrings}, w.l.o.g.\ $i=1$ and $j=B$.
Recall that, for any resolutions $g'\in \res(g)$ and $h'\in \res(h)$,
$\bigdiamond_{k=1}^{B} g_k'h_k'$ indicates whether $\repg{g'}>\repg{h'}$ ($10$), $\repg{g'}<\repg{h'}$
($01$), $g'=h'$ with $\parity(g')=0$ ($00$), or $g'=h'$ with $\parity(g')=1$
($11$). For $x=\metas \metas$, we must have that there are two pairs of
resolutions $g'$, $h'$ that result in (i) outputs $00$ and $11$, respectively,
or (ii) in outputs $01$ and $10$, respectively. It is straightforward to see
that this entails the claim (cf.~Table~\ref{table:validinputs}).
\end{proof}

We now prove the claim of the theorem by induction on $j-i+1$, i.e., the length of
the strings we feed to the operators. For $j=i$, we trivially have
$g_ih_i=\bigstarr \res(g_ih_i)$.

For the induction step, suppose $j>i$ and the claim holds for all shorter valid
strings. As, by Observation~\ref{obs:valid_substrings}, $g_{i,j}$ and
$h_{i,j}$ are valid strings, w.l.o.g.\ $i=1$ and $j=B$. Consider the
$\diamond_{\metas}$ operator (at the position between index $\ell$ and $\ell+1$)
on the left hand side that is evaluated last; we indicate this by parenthesis
and compute
\begin{align*}
&(g_1h_1\diamond_{\metas}\ldots\diamond_{\metas}g_{\ell}h_{\ell})\diamond_{\metas}
(g_{\ell+1}h_{\ell+1}\diamond_{\metas}\ldots\diamond_{\metas}g_Bh_B)\\
=\, &\left(\bigstarr \bigdiamond_{k=1}^{\ell}
\res(g_kh_k)\right)\diamond_{\metas}
\left(\bigstarr \bigdiamond_{k=\ell+1}^B \res(g_kh_k)\right)\\
\stackrel{\eqref{eq:diamondM}}{=} & \bigstarr \left(\res\left(\bigstarr
\bigdiamond_{k=1}^{\ell} \res(g_kh_k)\right)\diamond
\res\left(\bigstarr \bigdiamond_{k=\ell+1}^B \res(g_kh_k)\right)\right)\\
=\,&\bigstarr \left(\res(a)\diamond \res(b)\right)=:x,
\end{align*}
where $a=\bigstarr \bigdiamond_{k=1}^{\ell}\res(g_kh_k)$ and
$b=\bigstarr \bigdiamond_{k=\ell+1}^B \res(g_kh_k)$.

By the induction hypothesis, $a$ and $b$ do not depend on the order of
evaluation of the $\diamond_{\metas}$ operators. Thus, it suffices to show that
$x$ equals the right hand side of Equality~\eqref{eq:claimed_equality}.

We distinguish three cases. The first is that the right hand side of
\eqref{eq:claimed_equality} evaluates to $\metas \metas$. Then, by
Observation~\ref{obs:mm}, there is a (unique) index $m$ so that $g_m=h_m=\metas$
and $g_{1,m}=h_{1,m}$. If $m\leq \ell$, we have (again by
Observation~\ref{obs:mm}) that $a=\metas \metas$, i.e., $\res(a)=\{0,1\}^2$.
Checking Table~\ref{tab:diamond}, we see that each column contains both $01$ and
$10$. Hence, regardless of $b$, $x=\metas\metas$. On the other hand, if
$m>\ell$, then $a\in \{00,11\}$ and $b=\metas \metas$. Checking the $00$
and $11$ rows of Table~\ref{tab:diamond}, both of them contain $01$ and $10$,
implying that $x=\metas\metas$.

The second case is that the right hand side of~\eqref{eq:claimed_equality} does
not evaluate to $\metas \metas$, but $b=\metas\metas$. Then, by
Observation~\ref{obs:mm} and the fact that $g$ and $h$ are valid strings,
$g_{1,\ell},h_{1,\ell}\in \{0,1\}^\ell$ and $g_{1,\ell}\neq h_{1,\ell}$.
W.l.o.g., assume $\repg{g_{1,\ell}}> \repg{h_{1,\ell}}$. Then
$a=g_{1,\ell}\diamond h_{1,\ell}=10$ and the state machine given in
Figure~\ref{fig:fsm} determines output $10$ for inputs $g'\in \res(g)$ and
$h'\in \res(h)$. As the FSM outputs $\bigdiamond_{k=1}^B g_k'h_k'$, we conclude
that
\begin{align*}
\bigstarr \bigdiamond_{k=1}^B \res(g_kh_k)
&=\bigstarr_{g'h'\in \res(gh)}\left\{\bigdiamond_{k=1}^B g'_kh'_k \right\}\\
&=\bigstarr_{g'h'\in \res(gh)}\{01\}=01
\end{align*}
as well. Checking the $10$ row of Table~\ref{tab:diamond}, we see that $x=10$,
too, regardless of $b$.

The third case is that the right hand side of~\eqref{eq:claimed_equality} does
not evaluate to $\metas \metas$ and $b\neq \metas\metas$. By Observation~\ref{obs:mm}, also $a\neq
\metas \metas$. Accordingly, $|\res(a)|,|\res(b)|\in \{1,2\}$. We claim that this implies that
\begin{align*}
\res(a)&=\bigdiamond_{k=1}^{\ell}\res(g_kh_k), \quad
\res(b)=\bigdiamond_{k=\ell+1}^{B}\res(g_kh_k).
\end{align*}
This can be seen by noting that, for any set $S\subseteq \{0,1\}^2$, (i)
$\res(\bigstarr S)\supseteq S$ and (ii) $|\res(\bigstarr S)|=2$ necessitates
that $|S|\geq 2$, as otherwise $\bigstarr S\in \{0,1\}^2$ and thus
$|\res(\bigstarr S)|=1$. We conclude that
\begin{align*}
x&=\bigstarr \left(\res(a)\diamond \res(b)\right)\\
&=\bigstarr \left(\left(\bigdiamond_{k=1}^{\ell}\res(g_kh_k)\right)\diamond
\left(\bigdiamond_{k=\ell+1}^{B}\res(g_kh_k)\right)\right)\\
&=\bigstarr_{\substack{g'\in \res(g)\\h'\in \res(h)}}
\left(\left(\bigdiamond_{k=1}^{\ell}g_k'h_k'\right)\diamond
\left(\bigdiamond_{k=\ell+1}^{B}g_k'h_k'\right)\right)\\
&=\bigstarr_{\substack{g'\in \res(g)\\h'\in \res(h)}}
\left(\bigdiamond_{k=1}^{B}g_k'h_k'\right)=\bigstarr \left(\bigdiamond_{k=1}^B \res(g_kh_k)\right),
\end{align*}
as desired.
\end{proof}

We remark that we did \emph{not} prove that $\diamond_{\metas}$ is an
associative operator, just that it behaves associatively when applied to input
sequences given by valid strings.  Moreover, in general the closure of an
associative operator needs not be associative. A counter-example is given by
binary addition modulo $4$:
\begin{align*}
(0\metas+_{\metas}01)+_{\metas}01 = \metas\metas \neq 1\metas = 0\metas
+_{\metas} (01+_{\metas}01).
\end{align*}
Since $\diamond_{\metas}$ behaves associatively when applied to input
sequences given by valid strings, we can apply the results by~\cite{ladner1980parallel} on
parallel prefix computation to any implementation of $\diamond_{\metas}$.

\subsection{Obtaining the Outputs from $s^{(i)}_{\metas}$}\label{sec:outputdet}
Denote by $\outt\colon \{0,1\}^2\times \{0,1\}^2\to \{0,1\}^2$ the
operator given in Table~\ref{tab:output} computing $\rgmax\{g,h\}_i
\rgmin\{g,h\}_i$ out of $s^{(i-1)}$ and $g_ih_i$. The following theorem shows
that, for valid inputs, it suffices to implement $\outt_{\metas}$ to determine
$\rgmaxM\{g,h\}_i$ and $\rgminM\{g,h\}_i$ from $s^{(i-1)}_{\metas}$, $g_i$, and
$h_i$.
\begin{theorem}\label{thm:out}
Given valid inputs $g,h\in \validrg{B}$ and $i\in [B]$, it holds that
$$\outt_{\metas}(s^{(i-1)}_{\metas},g_ih_i) = \rgmaxM\{g,h\}_i \rgminM\{g,h\}_i.$$
\end{theorem}
\begin{proof}
By definition of $s^{(i-1)}_{\metas}$,
$\outt_{\metas}(s^{(i-1)}_{\metas},g_ih_i)$ does not depend on bits
$i+1,\ldots,B$. As by Observation~\ref{obs:valid_substrings} $g_{1,i},h_{1,i}\in
\validrg{i}$, we may thus w.l.o.g.\ assume that $B=i$.
For symmetry reasons, it suffices to show the claim for the first output bit
$\outt_{\metas}(s^{(B-1)}_{\metas},g_Bh_B)_1$ only; the other cases are analogous.

Recall that for $g,h\in \{0,1\}^B$, $s^{(B-1)}_{\metas}=s^{(B-1)}$ is the state
of the state machine given in Figure~\ref{fig:fsm} before processing the last
bit. Hence,
\begin{align*}
&\,\outt_{\metas}(s^{(B-1)}_{\metas},g_Bh_B)_1=\outt(s^{(B-1)},g_Bh_B)_1\\
=&\,\rgmax\{g,h\}_B=\rgmaxM\{g,h\}_B.
\end{align*}
Our task is to prove this equality also for the case where $g$ or $h$ contain a
metastable bit.

Let $j$ be the minimum index such that $g_j=\metas$ or $h_j=\metas$. Again, for
symmetry reasons, we may assume w.l.o.g.\ that $g_j=\metas$; the case
$h_j=\metas$ is symmetric. If $\repg{g_{1,j-1}}\neq\repg{h_{1,j-1}}$, suppose
w.l.o.g.\ (the other case is symmetric) that
$\repg{g_{1,j-1}}>\repg{h_{1,j-1}}$. Then $s^{(j-1)}=10$ and the state machine
is in absorbing state. Thus, regardless of further inputs, we get that
$s^{(B-1)}_{\metas}=s^{(B-1)}=10$ and
\begin{equation*}
\outt_{\metas}(s^{(B-1)}_{\metas},g_Bh_B)_1=\outt_{\metas}(10,g_Bh_B)_1=g_B.
\end{equation*}
Hence, suppose that $g_{1,j-1}=h_{1,j-1}$; we consider the case that
$\parity(g_{1,j-1})=0$ first, i.e., $s^{(j-1)}=00=s^{(0)}$. By
Observation~\ref{obs:valid_substrings}, $g_{j,B},h_{j,B}\in \validrg{B-j+1}$ and
thus w.l.o.g.\ $j=1$. If $B=1$,
\begin{align*}
\outt_{\metas}(s^{(B-1)}_{\metas},g_Bh_B)_1&=\outt_{\metas}(00,\metas h_B)_1
=\begin{cases}
1 & \mbox{if }h_1=1\\
\metas & \mbox{otherwise,}
\end{cases}
\end{align*}
which equals $\rgmaxM\{g,h\}_B$ (we simply have a $1$-bit code). If $B>1$, the
above implies that $g_{2,\ldots,B}=10\ldots0$, as the front bit of the code
changes only once, with $10\ldots0$ being the other bits (cf.\
Table~\ref{table:validinputs}). We distinguish several cases.
\begin{itemize}
\item [$h_1=\metas$:] Then also $h_{2,\ldots,B}=10\ldots 0$. Therefore
$g_B=h_B$, $\outt(s,g_Bh_B)_1=g_B=h_B$ for any $s\in \{0,1\}^2$, and
\begin{equation*}
\outt_{\metas}(s^{(B-1)}_{\metas},g_Bh_B)_1=g_B=h_B=\rgmaxM\{g,h\}_B.
\end{equation*}
\item [$h_1=1$ and $B=2$:] Note that $g$ is smaller than $h$ w.r.t.\ the total
order on valid strings (cf. Table~\ref{table:validinputs}), i.e., we need to
output $h_B=\rgmaxM\{g,h\}_B$. Consider the two resolutions of $g$, i.e., $01$
and $11$. If the first bit of $g$ is resolved to $0$, we end up with
$s^{(B-1)}=s^{(1)}=01$. If it is resolved to $1$, then $s^{(1)}=11$. Thus,
\begin{align*}
\outt_{\metas}(s^{(B-1)}_{\metas},g_Bh_B)_1
&=\outt_{\metas}(01,1h_B)_1*\outt_{\metas}(11,1h_B)_1\\
&=\bigstarr \bigcup_{h'\in \res(h)}\{h_B',\min\{1,h_B'\}\}\\
&=\bigstarr_{h'\in \res(h)} \{h_B'\}=h_B.
\end{align*}
\item [$h_1=1$ and $B>2$:] Again, $h_B=\rgmaxM\{g,h\}_B$. Consider the two
resolutions of $g$, i.e., $010\ldots0$ and $110\ldots0$. If the first bit of $g$
is resolved to $0$, we end up with $s^{(B-1)}=s^{(1)}=01$, as $01$ is an
absorbing state. If it is resolved to $1$, then $s^{(1)}=11$. As
$g_{2,\ldots,B}=10\ldots 0$, for any $h'\in \res(h)$, the state machine will end
up in either state $00$ (if $h'_{2,\ldots,B}=10\ldots 0$) or state $01$.
Overall, we get that (i) $s^{(B-1)}_{\metas}=01$, (ii) $s^{(B-1)}_{\metas}=00*01=0\metas$ and
$h_{2,\ldots,B}=1\ldots 0$, or (iii) $s^{(B-1)}_{\metas}=0\metas$ and
$h_{2,\ldots,B}=1\ldots 0\metas$ (cf.~Table~\ref{table:validinputs}). If (i)
applies, $\outt_{\metas}(s^{(B-1)}_{\metas},g_Bh_B)_1=h_B$. If (ii) applies,
$\outt_{\metas}(s^{(B-1)}_{\metas},g_Bh_B)_1=g_B=h_B$. If (iii) applies, then
\begin{align*}
\outt_{\metas}(s^{(B-1)}_{\metas},g_Bh_B)_1&=\outt_{\metas}(00,0\metas)_1*\outt_{\metas}(01,0\metas)_1\\
&=0*1*0*1=\metas = h_B.
\end{align*}
\item [$h_1=0$:] This case is symmetric to the previous two: depending on how
$g$ is resolved, we end up with $s^{(1)}=10$ or $s^{(1)}=00$, and need to
output $g_B$. Reasoning analogously, we see that indeed
$\outt_{\metas}(s^{(B-1)}_{\metas},g_Bh_B)_1=g_B$.
\end{itemize}
It remains to consider $\parity(g_{1,\ldots,j-1})=1$. Then
$s^{(j-1)}=11$. Noting that this reverses the roles of $\max$ and $\min$, we
reason analogously to the case of $\parity(g_{1,\ldots,j-1})=0$.
\end{proof}

\section{The Complete Circuit}\label{sec:complete}

Section~\ref{sec:optimal} breaks the task down to using the PPC framework to
compute $s^{(i)}_{\metas}$, $i\in [B]$, using $\diamond_{\metas}$ and then
$\outt_{\metas}$ to determine the outputs. Thus, we need to provide
implementations of $\diamond_{\metas}$ and $\outt_{\metas}$, and apply the
template from~\cite{ladner1980parallel}.

\subsection{Implementations of Operators}\label{sec:impop}
We provide optimized implementations based on fan-in $2$ $\ANDD$ and $\ORR$
gates and inverters here, cf.~Section~\ref{sec:model}. Depending on target
architecture and available libraries, more efficient solutions may be available.

\paragraph*{Implementing $\diamond_{\metas}$}
According to \cite{friedrichs16}, implementing $\diamond_{\metas}$ is possible,
and because $\diamond_{\metas}$ has constant fan-in and fan-out, it has constant
size.\\
We operate with the inverted first bits of the output of
$\diamond_{\metas}$. To this end, define $Nx:=\overline{x_1}x_2$ for $x\in
\{0,1,\metas\}^2$ and set
$$x \hatdiamM y := N(Nx\diamond_{\metas}Ny).$$
We compute $\bar{g}$ and work with inputs $\bar{g}$ and $h$ using operator
$\hat{\diamond}_{\metas}$. Theorem~\ref{thm:decompose} and elementary
calculations show that, for valid strings $g$ and $h$, we have
\begin{align*}
\left(\overline{g_1}h_1\hatdiamM\overline{g_2}h_2\right)
\hatdiamM\overline{g_3}h_3
 &= N(g_1h_1 \diamond_{\metas}g_2h_2)\hatdiamM\overline{g_3}h_3\\
 &= N((g_1h_1 \diamond_{\metas}g_2h_2)\diamond_{\metas}g_3h_3)\\
 &= N(g_1h_1 \diamond_{\metas}(g_2h_2\diamond_{\metas}g_3h_3))\\
 &= \overline{g_1}h_1\hatdiamM
\left(\overline{g_2}h_2\hatdiamM\overline{g_3}h_3\right),
\end{align*}
i.e., the order of evaluation of $\hatdiamM$ is insubstantial, just as
for $\diamond_{\metas}$. Moreover, as intended we get for all $1\leq i\leq j\leq
B$ that
\begin{align*}
\overline{g_i}h_i\hatdiamM\ldots\hatdiamM\overline{g_j}h_j
= N\left(g_ih_i\diamond_{\metas}\ldots\diamond_{\metas}g_jh_j\right).
\end{align*}
We concisely express operator $\diamond$ (Table~\ref{tab:output})
by the following logic formulas, where we already negate the first output bit.
\begin{align*}
  \overline{(s\diamond b)}_1 &=\overline{s_1}
  \cdot \left(s_2+\overline{b_1}\right)+s_2\cdot b_1\\
  (s\diamond b)_2 &=\overline{s_1}\cdot (s_2+b_2)
  +s_2\cdot\overline{b_2}
\end{align*}
This gives rise to depth-$3$ circuits containing in total $4$ $\ANDD$ gates, $4$
$\ORR$ gates, and $2$ inverters.\footnote{In the
base case, where $b_1=g_i$ for some $i\in \{1,\ldots,B\}$, we can save an
additional inverter.} From the gate behavior specified in Table~\ref{tab:gates},
one can readily verify that the circuit also implements
$\hat{\diamond}_{\metas}$ correctly.\footnote{Note that this is not true for
arbitrary logic formulas evaluating to $\overline{s \diamond b}$; e.g.,
$\overline{s\diamond b}_1=\left(\overline{s_1}+b_1\right)
  \cdot \left(s_2+\overline{b_1}\right)$, but the
corresponding circuit outputs $\metas\neq \overline{(10\diamond\metas
0)}_1=0$ for inputs $s=10$ and $b=\metas 0$.}
Since these circuits are identical to the ones used to
compute $\outt_{\metas}$, we give the implementation of such a selecting
circuit once in Figure~\ref{fig:select} and describe how to use it
in Table~\ref{tab:selection}. We remark that with identical select bits ($sel_1
= sel_2$), this circuit implements a $\cmux$ (a $\mux_{\metas}$ in our
terminology) as defined in~\cite{friedrichs16}.

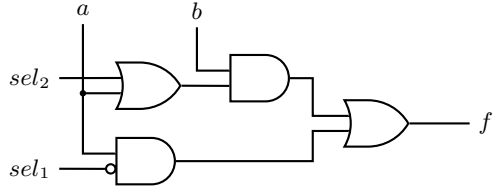
\begin{figure}
\centering\small
\begin{tikzpicture}[circuit logic US, every circuit symbol/.style={thick}]
	\node (i1) at (0.7, 2) {$a$};
	\node (i0) at (2.2, 2) {$b$};
	\node (i2) at (0, -.1)  {$sel_1$};
	\node (i3) at (0, 1.1) {$sel_2$};
	\node (o0) at (6, .5) {$f$};
	\node[or gate, inputs={nn}] (or1) at (1.5,1) {};
	\node[and gate, inputs={ni}] (and1) at (1.5,0) {};
	\node[and gate, inputs={nn}] (and2) at (3,1.1) {};
	\node[or gate, inputs={nn}] (or2)  at (4.5,.5) {};
        \draw[thick] (i3) -- (or1.input 1);
        \filldraw (.7, .9) circle (1pt);
        \draw[thick] (i1) ++(down:1.1) -| (or1.input 2);
        \draw[thick] (i1) -- ++(down:1.9) -| (and1.input 1);
        \draw[thick] (i2) -- (and1.input 2);
        \draw[thick] (or1.output) -- (and2.input 2);
        \draw[thick] (i0) -- ++(down:.8) -| (and2.input 1);
        \draw[thick] (and2.output) -- ++(right:.3) -- ++(down:.5) -| (or2.input 1);
        \draw[thick] (and1.output) -- ++(right:1.8) -- ++(up:.4) -| (or2.input 2);
        \draw[thick] (or2.output) -- (o0);
\end{tikzpicture}
\caption{Selection circuit with inputs: $a$, $b$, $sel_1$, $sel_2$ and output
$f$, used to implement $\hatdiamM$ and $\outt_{\metas}$.}
\label{fig:select}
\end{figure}

\begin{table}
\centering
\begin{tabular}{|c |c |c |c ||c |}
	\hline
	$sel_1$ & $sel_2$ & $a$ & $b$ & $f$\\ \hline \hline \rule{0pt}{8pt}
	$\overline{b_1}$ & $\overline{b_1}$ & $s_2$ & $\overline{s_1}$ & $(s \hatdiamM b)_1$\\
	$b_2$ & $b_2$ & $s_2$ & $\overline{s_1}$ & $(s \hatdiamM b)_2$\\
	$\overline{s_1}$ & $s_2$ & $b_1$ & $b_2$ & $\outt_{\metas}(s,b)_1$\\
	$s_2$ & $\overline{s_1}$ & $b_2$ & $b_1$ & $\outt_{\metas}(s,b)_2$ \\
	\hline
\end{tabular}
\caption{Connections to a selection circuit to compute the respective
operator $f$.}\label{tab:selection}
\end{table}

\paragraph*{Implementing $\outt_{\metas}$}
According to~\cite{friedrichs16}, $\outt_{\metas}$ can be implemented by a
circuit in our model; as $\outt$ has constant fan-in and fan-out, the circuit
has constant size.\\
The multiplication table of $\outt_{\metas}$, which is equivalent to
Table~\ref{tab:output}, is given in Table~\ref{tab:outt}.\\
We can concisely express the output function given in Table~\ref{tab:outt} by the
following logic formulas.
\begin{align*}
  \outt(s,b)_1 &= (\overline{s_1}+b_1)\cdot b_2
  +\overline{s_2}\cdot b_1\\
  \outt(s,b)_2 &= s_1\cdot b_2 +(s_2+b_2)\cdot b_1
\end{align*}
As mentioned before, instead of computing $s_1$, we determine and use as
input $\overline{s_1}$. Thus, the above formulas give rise to depth-$3$ circuits
that contain in total $4$ $\ANDD$ gates, $4$ $\ORR$ gates, and $2$ inverters (see
Figure~\ref{fig:select} and Table~\ref{tab:selection}); in fact, the circuit is
identical to the one used for $\hatdiamM$ with different inputs. From the gate behavior
specified in Table~\ref{tab:gates}, one can readily verify that the circuit indeed also
implements $\outt_{\metas}$.

\subsection{Implementation of $s^{(i)}_{\metas}$}\label{sec:Si}
We make use of the Parallel Prefix Computation (PPC)
framework~\cite{ladner1980parallel} to efficiently compute $s^{(i)}_{\metas}$ in
parallel for all $i \in [B]$. This framework requires an associative operator
$OP$. In our case, $OP=\hatdiamM$, which by Theorem~\ref{thm:decompose} is
associative on all relevant inputs. Given an implementation of $OP$, the circuit
is recursively constructed as shown in Figure~\ref{fig:ppcodd}, where the base
case $n=1$ is trivial. For $n$ that is a power of $2$, the depth and gate
counts are given as~\cite{even2006teaching}
\begin{equation}\label{eq:ppccostdelay}
\begin{aligned}
	\delay(PPC_{OP}(n)) &= (2 \log_2 n -1)\cdot \delay(OP),\\
	\cost(PPC_{OP}(n)) &= (2n- \log_2 n -2)\cdot \cost(OP)\:.
\end{aligned}
\end{equation}
\ifnum\dayte=0
For example for $n=4$, see Figure~\ref{fig:ppcfour}.
\fi

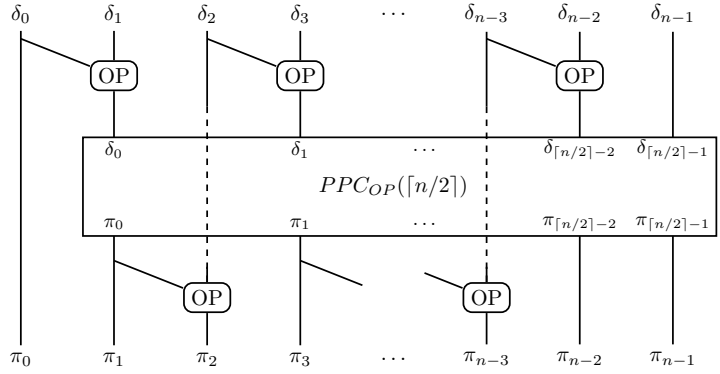
\begin{figure}\centering
\resizebox {.6\columnwidth} {!} {
\begin{tikzpicture}
\node (in1) at (0,5.6) {$\delta_0$};
\node (in2) at (1.5,5.6) {$\delta_1$};
\node (in3) at (3,5.6) {$\delta_2$};
\node (in4) at (4.5,5.6) {$\delta_3$};
\node (in5) at (6,5.6) {$\hdots$};
\node (in6) at (7.5,5.6) {$\delta_{n-3}$};
\node (in7) at (9,5.6) {$\delta_{n-2}$};
\node (in8) at (10.5,5.6) {$\delta_{n-1}$};

\node[rectangle, draw, rounded corners, thick] (op1) at (1.5,4.6) {OP};
\node[rectangle, draw, rounded corners, thick] (op2) at (4.5,4.6) {OP};
\node[rectangle, draw, rounded corners, thick] (op3) at (9,4.6) {OP};

\draw[thick] (1,2) rectangle (11.2,3.6);
\node at (6,2.8) {$PPC_{OP}(\lceil n/2\rceil)$};
\node at (1.5,3.4) {\small$\delta_{0}$};
\node at (4.5,3.4) {\small$\delta_{1}$};
\node at (6.5,3.4) {\small$\hdots$};
\node at (9,3.4) {\small$\delta_{\lceil n/2\rceil-2}$};
\node at (10.5,3.4) {\small$\delta_{\lceil n/2\rceil-1}$};
\node at (1.5,2.2) {\small$\pi_{0}$};
\node at (4.5,2.2) {\small$\pi_{1}$};
\node at (6.5,2.2) {\small$\hdots$};
\node at (9,2.2) {\small$\pi_{\lceil n/2\rceil-2}$};
\node at (10.5,2.2) {\small$\pi_{\lceil n/2\rceil-1}$};

\node[rectangle, draw, rounded corners, thick] (op4) at (3,1) {OP};
\node[rectangle, draw, rounded corners, thick] (op5) at (7.5,1) {OP};

\node (out1) at (0,0) {$\pi_0$};
\node (out2) at (1.5,0) {$\pi_1$};
\node (out3) at (3,0) {$\pi_2$};
\node (out4) at (4.5,0) {$\pi_3$};
\node (out5) at (6,0) {$\hdots$};
\node (out6) at (7.5,0) {$\pi_{n-3}$};
\node (out7) at (9,0) {$\pi_{n-2}$};
\node (out8) at (10.5,0) {$\pi_{n-1}$};

\draw[thick] (in1) -- (out1);
\draw[thick] (in1) ++(down:.4) -- (op1);
\draw[thick] (in2) -- (op1);
\draw[thick] (in3) ++(down:.4) -- (op2);
\draw[thick] (in4) -- (op2);
\draw[thick] (in6) ++(down:.4) -- (op3);
\draw[thick] (in7) -- (op3);
\draw[thick] (in8) -- ++(down:2);
\draw[thick] (in3) -- ++(down:1.5);
\draw[thick] (in6) -- ++(down:1.5);
\draw[thick] (op4) -- ++(up:.5);
\draw[thick] (op5) -- ++(up:.5);
\draw[thick, dashed] (in3) ++(down:1.5) -- (op4) ++(up:.5);
\draw[thick, dashed] (in6) ++(down:1.5) -- (op5) ++(up:.5);
\draw[thick] (op1) -- ++(down:1);
\draw[thick] (op2) -- ++(down:1);
\draw[thick] (op3) -- ++(down:1);
\draw[thick] (out2) -- ++(up:2);
\draw[thick] (out4) -- ++(up:2);
\draw[thick] (out7) -- ++(up:2);
\draw[thick] (out8) -- ++(up:2);
\draw[thick] (1.5,2) ++(down:.4) -- (op4);
\draw[thick] (4.5,2) ++(down:.4) -- (5.5,1.2);
\draw[thick] (6.5,1.4) -- (op5);
\draw[thick] (op4) -- (out3);
\draw[thick] (op5) -- (out6);
\end{tikzpicture}
}
\caption{Recursive construction of $PPC_{OP}(n)$ for odd $n$, computing $\pi_i =
\delta_0 OP \ldots OP \delta_i$. For even $n$ the rightmost input ($\delta_{n-1}$) and output ($\pi_{n-1}$) are
not present. Dashed lines are not connected to $PPC_{OP}(\lceil n/2\rceil)$.}
\label{fig:ppcodd}
\end{figure}

\ifnum\dayte=0
\begin{figure}\centering
\begin{tikzpicture}
\node (in1) at (0,3) {$\delta_0$};
\node (in2) at (1.5,3) {$\delta_1$};
\node (in3) at (3,3) {$\delta_2$};
\node (in4) at (4.5,3) {$\delta_3$};

\node[rectangle, draw, rounded corners, thick] (op1) at (1.5,2) {OP};
\node[rectangle, draw, rounded corners, thick] (op2) at (3,.8) {OP};
\node[rectangle, draw, rounded corners, thick] (op3) at (4.5,2) {OP};
\node[rectangle, draw, rounded corners, thick] (op4) at (4.5,1) {OP};

\draw[thick] (in1) -- (out1);
\draw[thick] (in1) ++(down:.4) -- (op1);
\draw[thick] (in2) -- (op1);
\draw[thick] (op1) ++(down:.4) -- (op2);
\draw[thick] (in3) -- (op2);
\draw[thick] (op1) -- (out2);
\draw[thick] (op2) -- (out3);
\draw[thick] (in4) -- (op3);
\draw[thick] (in3) ++(down:.4) -- (op3);
\draw[thick] (op3) -- (op4);
\draw[thick] (op1) ++(down:.4) -- (op4);
\draw[thick] (op4) -- (out4);
\filldraw (0, 2.6) circle (1pt);
\filldraw (3, 2.6) circle (1pt);
\filldraw (1.5, 1.6) circle (1pt);

\node (out1) at (0,0) {$\pi_0$};
\node (out2) at (1.5,0) {$\pi_1$};
\node (out3) at (3,0) {$\pi_2$};
\node (out3) at (4.5,0) {$\pi_3$};
\end{tikzpicture}
\caption{Construction of a four input $PPC_4(OP)$\moti{make sure that the notations match the captions.}}
\label{fig:ppcfour}
\end{figure}
\fi
\subsection{Putting it All Together}
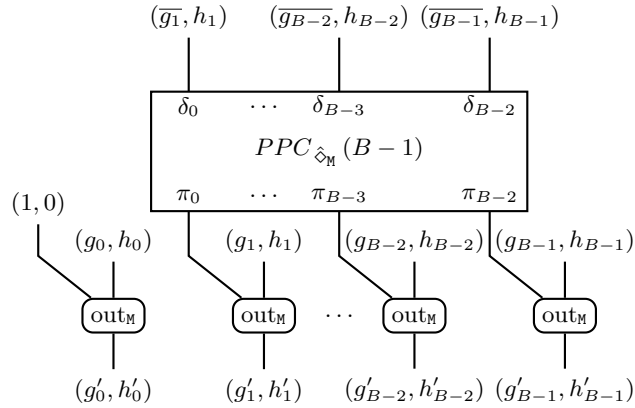
\begin{figure}\centering
\begin{tikzpicture}[scale=0.99]\small
\node (in1) at (0,5) {$(\overline{g_1},h_1)$};
\node (in2) at (2,5) {$(\overline{g_{B-2}},h_{B-2})$};
\node (in3) at (4,5) {$(\overline{g_{B-1}},h_{B-1})$};

\draw[thick] (-.5,2.4) rectangle (4.5,4);
\node at (2,3.2) {$PPC_{\mathlarger{\hatdiamM}}(B-1)$};
\node at (0,3.8) {\small$\delta_{0}$};
\node at (1,3.8) {\small$\hdots$};
\node at (2,3.8) {\small$\delta_{B-3}$};
\node at (4,3.8) {\small$\delta_{B-2}$};
\node at (0,2.6) {\small$\pi_{0}$};
\node at (1,2.6) {\small$\hdots$};
\node at (2,2.6) {\small$\pi_{B-3}$};
\node at (4,2.6) {\small$\pi_{B-2}$};

\node (zero) at (-2,2.5) {$(1,0)$};

\node (in5) at (-1,2) {$(g_0,h_0)$};
\node (in6) at (1,2) {$(g_1,h_1)$};
\node (in7) at (3,2) {$(g_{B-2},h_{B-2})$};
\node (in8) at (5,2) {$(g_{B-1},h_{B-1})$};

\node[rectangle, draw, rounded corners, thick] (sel1) at (-1,1) {$\outt_{\metas}$};
\node[rectangle, draw, rounded corners, thick] (sel2) at (1,1) {$\outt_{\metas}$};
\node at (2,1) {$\hdots$};
\node[rectangle, draw, rounded corners, thick] (sel3) at (3,1) {$\outt_{\metas}$};
\node[rectangle, draw, rounded corners, thick] (sel4) at (5,1) {$\outt_{\metas}$};

\node (out1) at (-1,0) {$(g^\prime_0,h^\prime_0)$};
\node (out2) at (1,0) {$(g^\prime_1,h^\prime_1)$};
\node (out3) at (3,0) {$(g^\prime_{B-2},h^\prime_{B-2})$};
\node (out4) at (5,0) {$(g^\prime_{B-1},h^\prime_{B-1})$};

\draw[thick] (in1) -- ++(down:1);
\draw[thick] (in2) -- ++(down:1);
\draw[thick] (in3) -- ++(down:1);
\draw[thick] (in5) -- (sel1);
\draw[thick] (in6) -- (sel2);
\draw[thick] (in7) -- (sel3);
\draw[thick] (in8) -- (sel4);
\draw[thick] (zero) -- (-2,1.8) -- (sel1);
\draw[thick] (0,2.4) -- (0,1.8) -- (sel2);
\draw[thick] (2,2.4) -- (2,1.8) -- (sel3);
\draw[thick] (4,2.4) -- (4,1.8) -- (sel4);
\draw[thick] (sel1) -- (out1);
\draw[thick] (sel2) -- (out2);
\draw[thick] (sel3) -- (out3);
\draw[thick] (sel4) -- (out4);
\end{tikzpicture}
\caption{Construction of $\twosort(B)$ from $out_{\metas}$ and
$PPC_{\hatdiamM}(B-1)$. $PPC_{OP}(B-1)$ is specified in Figure~\ref{fig:ppcodd},
and we use the implementations of $\hatdiamM$ and $\outt_{\metas}$ specified by
Figure~\ref{fig:select} and Table~\ref{tab:selection}. For input
$Ns^{(0)}=(1,0)$, $\outt_{\metas}$ reduces to an $\ANDD$ and an $\ORR$ gate.}
\label{fig:sortppc}
\end{figure}

\begin{theorem}
  The circuit depicted in Figure~\ref{fig:sortppc} implements $\twosort(B)$ according to Definition~\ref{def:twosort}. Its delay
  is $O(\log B)$ and its gate count is $O(B)$.
\end{theorem}
\begin{proof}
Theorem~\ref{thm:out} implies that we can compute the output by feeding
$s^{(i-1)}_{\metas}$ and $g_ih_i$, for $i\in \{1,\ldots,B\}$, into a circuit
computing $\outt_{\metas}$. We determine $s^{(i-1)}_{\metas}$ as discussed in
Section~\ref{sec:Si}, which is feasible by Theorem~\ref{thm:decompose}. We use
the implementations of $\hatdiamM$ and $\outt_{\metas}$ given in
Section~\ref{sec:impop}, cf.~Figure~\ref{fig:select} and Table~\ref{tab:selection},
respectively. As these circuits have constant depth and gate count, the overall
complexity bounds immediately follow from~\eqref{eq:ppccostdelay}.
\end{proof}

\section{Simulation Results}\label{sec:simulation}

\paragraph*{Design Flow}
Our design flow makes use of the following tools:
\begin{inparaenum}[(i)]
  \item design entry: Quartus,
  \item behavioral simulation: ModelSim,
  \item synthesis: Encounter RTL Compiler (part of Cadence tool set) with NanGate $45$\,nm Open Cell Library,
  \item place \& route: Encounter (part of Cadence tool set) with NanGate $45$\,nm Open Cell Library.
\end{inparaenum}
\paragraph*{Design Flow adaptations for MC}
During synthesis the VHDL description of a circuit is automatically mapped to standard cells provided by a standard cell
library. The standard cell library used for the experiments provides besides simple $\ANDD$, $\ORR$ or Inverter gates
also more powerful $\AOI$ (And-Or-Invert) gates, which combine multiple boolean connectives and optimize them
on transistor level. Since we did not analyze
the behaviour of more complex $\AOI$ gates in face of metastability, we restrict our implementation to use only
$\ANDD$, $\ORR$ and Inverter gates. To ensure this, we performed the mapping to standard cells by hand.
The following standard cells have been used to map the logic gates to hardware:
\begin{inparaenum}[(i)]
  \item INV\_X1: Inverter gate,
  \item AND2\_X1: $\ANDD$ gate,
  \item OR2\_X1: $\ORR$ gate.
\end{inparaenum}
In the documentation of the NanGate $45$\,nm Open Cell Library it can be seen
that these cells in fact compute the metastable closure of the respective Boolean connective.

After mapping the design by hand, we can disable the optimization in the
synthesis step and go on with place and route. This prevents the RTL Compiler from performing Boolean
optimization on the design, which may destroy the MC properties of our circuits.

\begin{table}
\begin{center}
\begin{tabular}{| c || c || c | c | c |}
\hline 
 $B$ & Circuit   & \# Gates & Area \scalebox{.8}{$[\mu
 $m$^{\scalebox{.6}{$2$}}]$}  & Delay \scalebox{.8}{[ps]}
\tabularnewline
\hline
\hline
\multirow{3}{*}{$B=2$} & This paper & 13        & 17.486                   & 119
\tabularnewline
 & \cite{date17} & 34        & 49.42                    & 268
 \tabularnewline
 & $\binary$ & 8 & 15.582 &145
\tabularnewline
\hline
\multirow{3}{*}{$B=4$}& This paper  & 55        & 73.752                    & 362
\tabularnewline
 & \cite{date17} & 160        & 230.3                    & 498
 \tabularnewline
 & $\binary$ & 19 & 34.58 &288
\tabularnewline
\hline
\multirow{3}{*}{$B=8$}& This paper & 169        & 227.29                   & 516
\tabularnewline
 & \cite{date17}  & 504        & 723.52                   & 827
 \tabularnewline
 & $\binary$ & 41 & 73.752 &477
\tabularnewline
\hline
\multirow{3}{*}{$B=16$}& This paper & 407       & 548.016                    & 805
\tabularnewline
 & \cite{date17}  & 1344       & 1928.262                    & 1233
\tabularnewline
 & $\binary$ & 81 & 151.648 &422
\tabularnewline
\hline
\end{tabular}
\end{center}
\caption{Comparison of gate count, delay, and area of $\twosort(B)$ from
this paper and~\cite{date17}, and $\binary$, an optimized comparator taking
binary inputs.}\label{table:summary}
\end{table}

\begin{table*}[b]\scriptsize
\begin{center}
\begin{adjustbox}{center}
\begin{tabular}{| c | c ||  c | c | c || c| c |c || c| c | c || c | c | c |}
\hline
\multirow{2}{*}{$B$} & \multirow{2}{*}{Circuit}   & \multicolumn{3}{|c ||}{$\foursort$} & \multicolumn{3}{|c ||}{$\sevensort$} & \multicolumn{3}{|c ||}{$\tensortc$} & \multicolumn{3}{|c |}{$\tensortd$}\\
                 \cline{3-14}
     &           & gates & area & delay & gates & area & delay & gates & area & delay & gates & area & delay
\tabularnewline
\hline
\hline
\multirow{3}{*}{$2$}& here & 65 & 87.402 & 357 & 208 & 279.741 & 714 & 377 & 506.912 & 912 & 403 & 541.968 & 833
\tabularnewline
                      & \cite{date17} & 170 & 247.016 &846 & 544 & 790.44 &1715 & 986 & 1432.62 &2285 & 1054 & 1531.467 &2010
                      \tabularnewline
                      & $\binary$ & 40 & 77.91 &478 & 128 & 249.326 &953 & 232 & 451.815 &1284 & 248 & 483 &1145
\tabularnewline
\hline
\multirow{3}{*}{$4$}& here & 275 & 368.641 & 640 & 880 & 1179.528 & 1014 & 1595 & 2137.905 & 1235 & 1705 & 2285.514 & 1133
\tabularnewline
                      & \cite{date17} & 800 & 1151.472 &1558 &2560 & 3684.541 & 3147 & 4640 & 6678.294 &4207 & 4960 & 7138.74 &3681
                      \tabularnewline
                      & $\binary$ & 95 & 172.935 & 906 & 304 & 553.28 & 1810 & 551 & 1002.848 & 2429& 589 & 1072.099 & 2143
\tabularnewline
\hline
\multirow{3}{*}{$8$}& here & 845 & 1136.184 & 1396 & 2704 & 3636.08 & 1921 & 4901 & 6590.283 & 2179 & 5239 & 7044.541 & 2059
\tabularnewline
                      & \cite{date17} & 2520 & 3617.67 & 2394 & 8064 & 11576.32 & 4715 &14616 & 20982.542 & 6252& 15624 & 22429.176 & 5481
                      \tabularnewline
                      & $\binary$ & 205 & 368.641 &1475 & 656 & 1179.528 &2948 & 1189 & 2137.905 &3945& 1271 & 2285.514 &3470
\tabularnewline
\hline
\multirow{3}{*}{$16$}& here  & 2035 & 2739.961 & 2069  & 6512 & 8767.374 & 3396 & 11803 & 15891.12 & 4030 & 12617 & 16987.194 & 3844
\tabularnewline
                      & \cite{date17}  & 6720 & 9640.75 &3396  & 21504 & 30849.875 &6415 & 38976 & 55916.448 &8437& 41664 & 59772.132 &7458
                      \tabularnewline
                      & $\binary$ & 405 & 530.67 &1298 & 1296 & 2425.99 &2600 & 2349 & 4397.085 &3474 & 2511 & 4700.304 &3050
\tabularnewline
\hline
\end{tabular}
\end{adjustbox}
\end{center}
\caption{Simulation results for metastability-containing sorting networks with
$n\in \{4,7,10\}$ for $B$-bit inputs. $\tensortc$ optimizes gate count~\cite{codish2014twenty}, $\tensortd$ depth~\cite{bundala2014optimal}; for $n\in \{4,7\}$, the sorting networks are optimal w.r.t.\
both measures. Simulation results are: (i) number of gates, (ii) postlayout area $[\mu m^2]$ and
(iii) prelayout delay $[ps]$.}
\label{table:sorting}
\end{table*}
\paragraph*{The binary benchmark: $\binary$}
Following~\cite{date17}, we also compare our sorting networks to a standard
(non-containing!) sorting design. $\binary$ uses a simple VHDL statement to
compare both inputs:
\begin{lstlisting}[language=vhdl, frame=single, caption={{\tt VHDL code excerpt of
binary comparator}}, label={lst:comp}, firstline=18, lastline=22]
    if (a > b) then
      greater <= '1';
    else
      greater <= '0';
    end if;
\end{lstlisting}
Each output is connected to a standard multiplexer, where the signal $greater$ is used as the select
bit for both multiplexers.

The binary design follows a standard design flow, which uses the tools listed
above. In short, $\binary$ follows the same design process as
$\twosort$, but then undergoes optimization using a more powerful
set of basic gates.

We emphasize that the more powerful $\AOI$ gates combine multiple boolean
functions and optimize them on gate level, yet each of them is still counted as
one gate. Thus, comparing our design to the binary design in terms of gate
count, area, and delay disfavors our solution. Moreover, the optimization
routine switches to employing more powerful gates when going from $B=8$ to
$B=16$ (See Table~\ref{table:sorting}) resulting in a \emph{decrease} of the
delay of the binary implementation.

Nonetheless, our design performs comparably to the non-containing binary design
in terms of delay, cf.~Table~\ref{table:summary}. This is quite notable, as
further optimization on the transistor level or using more powerful $\AOI$
gates is possible, with significant expected gains. The same applies to gate
count and area, where a notable gap remains. Recall, however, that the binary
design hides complexity by using more advanced gates and does not contain
metastability.

We remark that we refrained from optimizing the design by making use of all
available gates or devising transistor-level implementations for two reasons.
First, such an approach is tied  to the utilized library or requires design of
standard cells. Second, it would have been unsuitable for a comparison
with~\cite{date17}, which does not employ such optimizations either.

\paragraph*{Comparison to State of the Art}

Our circuits show large improvements over~\cite{date17} in all performance
measures. Delays, gate counts, and area are all smaller by factors between
roughly $1.5$ and $3.5$. In particular, for $B=16$ delay is roughly cut in half,
while gate count and area decrease by factors of $3$ or more.

\section{Discussion}\label{sec:conclusion}

In this paper, we provide asymptotically optimal MC sorting primitives. We
achieve this by applying results on parallel prefix
computation~\cite{ladner1980parallel}, which requires to establish that the
involved operators behave associative on the relevant inputs. Our
circuits are purely combinational and are glitch-free (as they are MC). 
Compared to standard sorting networks, we roughly match delay, but fall behind
on gate count and area. However, we used gate-level implementations
of $\outt_{\metas}$ and $\diamond_{\metas}$ restricted to $\ANDD$ and $\ORR$
gates and inverters. Transistor-level implementations, which are a
straightforward optimization, would decrease size and delay of the derived
circuits further. We expect that this will result in circuits that perform on
par with standard sorting networks. In light of these properties, we believe our
circuits to be of wide applicability.

\pagebreak
\paragraph*{Acknowledgements}
This project has received funding from the European Research Council (ERC) under the European Union's Horizon 2020 research and
innovation programme (grant agreement 716562).

\bibliographystyle{plain}
\bibliography{comp_short}



\end{document}